\documentclass[letterpaper, 10 pt, conference]{ieeeconf} 

\IEEEoverridecommandlockouts                   
\usepackage{cite}
\usepackage{xcolor}
\usepackage[hidelinks]{hyperref}
\usepackage{amsmath,amssymb,amsfonts}
\usepackage{mathtools}
\usepackage{algorithm}
\usepackage{algorithmic}
\usepackage{subcaption}
\usepackage{graphicx}
\usepackage{textcomp}
\usepackage{booktabs}
\usepackage{enumerate}

\DeclareSymbolFont{bbold}{U}{bbold}{m}{n}
\DeclareSymbolFontAlphabet{\mathbbold}{bbold}

\newtheorem{definition}{Definition}[section]

\newtheorem{theorem}{Theorem}

\newtheorem{remark}{Remark}[section]
\newtheorem{problem}{Problem}
\usepackage{comment}
\usepackage{footnote}
\makesavenoteenv{algorithm}

\newcommand{\real}{\mathbb{R}}

\title{\LARGE \bf
Data-Driven Compositional Safety Verification of Interconnected Monotone Systems 
}

\author{Amirreza Alavi, \IEEEmembership{Graduate Student Member, IEEE}, Majid Zamani, \IEEEmembership{Senior Member, IEEE},\\ and Saber Jafarpour, \IEEEmembership{Member, IEEE}
\thanks{This manuscript is an extended version of a paper accepted for
presentation at the 65th IEEE Conference on Decision and Control
(CDC 2026) and for publication in its proceedings.}
\thanks{This work was supported by the NSF CAREER grant CNS-2145184.}
\thanks{A. Alavi, M. Zamani, and S. Jafarpour are with the University of Colorado Boulder, Boulder, CO 80309, (emails:{\tt\small \{seyedamirreza.alavi, majid.zamani, saber.Jafarpour\}@colorado.edu})}
}

\begin{document}

\maketitle
\begingroup
\renewcommand{\thefootnote}{}
\footnotetext{\textcopyright{} 2026 IEEE. Personal use of this material
is permitted. Permission from IEEE must be obtained for all other uses,
in any current or future media, including reprinting/republishing this
material for advertising or promotional purposes, creating new
collective works, for resale or redistribution to servers or lists, or
reuse of any copyrighted component of this work in other works.}
\endgroup
\thispagestyle{empty}
\pagestyle{empty}

\begin{abstract}
This paper introduces a sample-efficient compositional method for formal safety verification of interconnected monotone systems without requiring explicit models of the local subsystems.
Existing data-driven approaches either lack formal safety guarantees or rely on dense Lipschitz-based discretizations of the state space to provide such guarantees, which leads to significant computational overhead and limits scalability. 
In contrast, we leverage the monotonicity of local subsystems to construct tractable local interval-barrier certificates using only boundary evaluations of each subsystem in a decentralized manner, together with a global condition that guarantees the safety of the overall interconnected system. 
At the local level, our framework learns interval-barrier certificates using monotone neural networks from boundary samples induced by partitions of the local state and internal-input spaces.
At the global level, it composes these local neural interval-barrier certificates using the interconnection structure to certify the safety of the overall system. 
Furthermore, under appropriate structural assumptions, we reformulate the global safety condition into a scalable form that can be directly incorporated into the neural network loss. This enables the enforcement of overall system safety during local training.
The experimental results demonstrate the effectiveness and scalability of the proposed method.

\end{abstract}


\section{Introduction}
Complex networks of large-scale interconnected systems arise naturally in a wide range of applications, including traffic networks~\cite{FB-CTDS} and biological networks~\cite{EDS:07}. As these systems often operate in safety-critical environments, formal safety verification is essential. However, ensuring safety remains challenging due to their complexity and tight coupling. A common approach is to employ compositional methods that avoid treating the network as a monolithic system and instead derive network-level guarantees from the structure of subsystems and their interconnections.

In many real-world applications, the precise dynamics of individual components of the system are often unavailable due to physical complexity or environmental uncertainty, while the interconnection topology of the overall network is known. At the same time, rapid advances in big data management and low-cost distributed computing have enabled the availability of rich data samples for subsystems. These developments have led to a growing interest in data-driven approaches for system safety, e.g.,~\cite{sadraddini2018formal,10015033}.

\paragraph*{Related Work}

Classical approaches to safety verification often rely on barrier certificates to provide formal guarantees. A barrier certificate is a function whose level sets separate safe and unsafe regions, thereby certifying safety~\cite{prajna2004safety,ADA-XX-JWG-PT:17}. Existing methods construct such certificates using techniques such as sum-of-squares optimization~\cite{prajna2004safety} and SMT-based approaches~\cite{edwards2024fossil}. While effective, these approaches typically require precise system models and become computationally expensive in high-dimensional settings. More recently, neural networks have emerged as a promising alternative for learning barrier functions directly from data~\cite{zhao2020synthesizing,zhang2024exact}. However, these methods typically suffer from high sample complexity and limited scalability, particularly in large-scale interconnected systems where each subsystem may itself be high-dimensional.

To address these challenges, compositional approaches to safety verification have been developed in the literature. Existing methods can be broadly categorized into abstraction-based approaches, which construct finite-state approximations~\cite{meyer2017compositional}, and abstraction-free approaches, such as barrier certificates~\cite{sloth2012compositional}. However, when extended to the data-driven setting, compositional safety frameworks exhibit notable limitations: they either lack formal guarantees~\cite{noroozi2021data} or rely on Lipschitz bounds~\cite{samariTAC2026}, whose estimation is computationally demanding. Moreover, these approaches may still incur significant computational and data requirements due to the curse of dimensionality, especially when the subsystems themselves are high-dimensional.

Monotone systems exhibit highly ordered dynamical behaviors~\cite{DA-EDS:03} and arise naturally in a wide range of applications, including traffic networks~\cite{SC-MA:15}, population dynamics~\cite{leenheer2004predator}, and biological systems~\cite{EDS:07}. These structural properties have recently been leveraged to enable safety verification with improved sample efficiency (e.g.,~\cite{alaviCDC2025,FelipeCDC2025}).
However, such monotonicity-based approaches are not applicable when the overall interconnected system is non-monotone~\cite{enciso2006nonmonotone,de2005predator}. Moreover, even when monotonicity holds, their scalability remains limited in large-scale systems.

\paragraph*{Contributions}

In this work, we propose a novel data-driven and compositional framework for formal safety verification of networks of interconnected monotone subsystems, where the overall interconnection may not preserve monotonicity.
Our approach consists of (i) a \emph{local} search for certificates from sampled data for each subsystem, and (ii) a \emph{global} composition of these local certificates, leveraging the system topology to establish the safety of the overall system.

At the local level, we learn monotone neural interval-barrier functions from a finite set of samples. By exploiting the monotone structure, we reduce verification over continuous state spaces to localized boundary conditions, leading to a computationally efficient and sample-efficient framework. In contrast to existing methods, our approach does not rely on Lipschitz continuity assumptions or the estimation of Lipschitz constants, thereby avoiding additional assumptions and computational overhead.
At the global level, leveraging the interconnection structure of the system, we establish a Linear Matrix Inequality (LMI) that guarantees that the composition of the local barrier functions ensures the safety of the overall system. 
We design a tailored loss function for monotone neural architectures that enforces the interval-barrier conditions, ensuring correctness when the loss is zero. Moreover, under suitable conditions, the global compositional constraint can be embedded into local learning in a neural-friendly manner, eliminating the need to explicitly construct or verify the global LMI condition. 
Finally, we demonstrate the effectiveness and scalability of the proposed framework on benchmark examples. The simulation results highlight its scalability and sample efficiency.

\section{Preliminaries and Problem Statement}

\subsection{Notation}
\label{subsec:notation}
The identity matrix of dimension $n$ is denoted by $I_n$. The space of all $n\times n$ symmetric matrices is denoted by $\mathbb{S}^n$. Given matrices $A$ and $B$, $A \otimes B$ denotes their Kronecker product. For symmetric matrices, $A \preceq B$ indicates that $A-B$ is negative semidefinite. Given a matrix $A\in \real^{m\times n}$, the matrices $A^{+},A^{-}\in \real^{m\times n}$ contain the positive 
and negative entries of $A$, respectively, with zeros elsewhere, so that 
$A = A^{+} + A^{-}$. 
For vectors $x,y \in \real^n$, we denote $x\le y$ if $x_i\le y_i$, for every $i \in \{1,\ldots,n\}$. We represent a hyperrectangle by $[\underline{x},\overline{x}]= \{x\in \real^n\mid \underline{x}\le x \le \overline{x}\}$. We define $\mathrm{ReLU}$ by $\mathrm{ReLU}(x) {=}\max(0, x)$. Throughout the paper, $0_{n \times n}$ denotes the $n \times n$ zero matrix. For a matrix $A$, $A^\top$ denotes its transpose. For a differentiable function $g$, $Dg(x)$ denotes its Jacobian. A map $f:\mathcal{X}\to \mathcal{X}$ is said to be sign-stable on $\mathcal{X}$ if, for every $x \in \mathcal{X}$, each partial derivative $\frac{\partial f_i}{\partial x_j}(x)$ maintains a consistent sign; that is, either $\frac{\partial f_i}{\partial x_j}(x) > 0$ for all $x \in \mathcal{X}$, or $\frac{\partial f_i}{\partial x_j}(x) < 0$ for all $x \in \mathcal{X}$ or $\frac{\partial f_i}{\partial x_j}(x) = 0$ for all $x \in \mathcal{X}$.  For sets $\mathcal{X}_1,\dots,\mathcal{X}_m$, their Cartesian product is denoted by $\mathcal{X}_1 \times \cdots \times \mathcal{X}_m = \prod_{i=1}^m \mathcal{X}_i$.

\subsection{Interconnected Systems}

The class of discrete-time interconnected dynamical systems is defined as follows.

\begin{definition}[Interconnected systems]\label{def:dt_system}
An interconnected system $\Sigma$ is defined as the interconnection of $N$ discrete-time subsystems $\Sigma_i = (\mathcal{X}_i,\mathcal{X}_{0,i},\mathcal{W}_i,g_i)$, where $\mathcal{X}_i \subseteq \mathbb{R}^{n_i}$ is the state set, $\mathcal{X}_{0,i}\subseteq\mathcal{X}_i$ is the initial set, $\mathcal{W}_i \subseteq \mathbb{R}^{m_i}$ is the internal input set, and $g_i : \mathcal{X}_i \times \mathcal{W}_i \to \mathcal{X}_i$ is the state-transition map, for $i \in \{1,\dots,N\}$. The dynamics of the subsystem $\Sigma_i$ are described by
\begin{equation}\label{eq:sigmai}
\quad x_i^{+} = g_i(x_i, w_i).
\end{equation}
The subsystems are interconnected through 
$w = Mx$,
where $x = (x_1^{\top},\dots,x_N^{\top})^{\top}$, $w = (w_1^{\top},\dots,w_N^{\top})^{\top}$, and $M \in \mathbb{R}^{m \times n}$ is the interconnection matrix, with $n = \sum_{i=1}^{N} n_i$ and $m = \sum_{i=1}^{N} m_i$. The interconnected system is represented by $\Sigma = (\mathcal{X},\mathcal{X}_0,g,M)$, where $\mathcal{X} = \prod_{i=1}^{N} \mathcal{X}_i$ is the overall state set, $\mathcal{X}_0 = \prod_{i=1}^{N} \mathcal{X}_{0,i}$ is the overall initial set, and the map $g:\mathcal{X}\to \mathcal{X}$ given by $g(x) = (g_1(x_1,(Mx)_1)^{\top},\ldots, g_{N}(x_N,(Mx)_N)^{\top})^{\top}$ describes the overall evolution of the interconnected system.
\smallskip
\end{definition}

While the framework of interconnected systems is highly general, it is often too broad to enable scalable analysis without additional structure. In this paper, we exploit monotonicity as a key structural property at the subsystem level. Monotone systems exhibit order-preserving dynamics, which lead to significant simplifications in analysis.

\begin{definition}[Monotone Systems]
A discrete-time dynamical system  
$\Sigma = (\mathcal{X}, \mathcal{X}_{0}, \mathcal{W}, f)$ 
with state-transition map $f : \mathcal{X} \times \mathcal{W} \to \mathcal{X}$ 
is \emph{monotone} if
\[
x \le y \ \text{and}\ w \le v 
\;\;\implies\;\;
f(x,w) \le f(y,v).
\]
\end{definition}
\smallskip

Assume that $\mathcal{X} \subseteq \mathbb{R}^n$ and $\mathcal{W} \subseteq \mathbb{R}^m$. 
If the transition map $f$ is differentiable and the sets $\mathcal{X}$ and $\mathcal{W}$ are convex, 
then the system $\Sigma$ is monotone if and only if
$Df(x,w) \ge 0,\forall (x,w) \in \mathcal{X} \times \mathcal{W}$,
where $Df(x,w) \in \mathbb{R}^{n \times (n+m)}$ denotes the Jacobian of $f$ with respect to $(x,w)$, and the inequality is interpreted elementwise~\cite{DA-EDS:03}. In this paper, we focus on interconnected systems in which subsystems are monotone. 
\begin{definition}[Interconnected Monotone Systems]\label{def:IMS}
  An interconnected system $\Sigma = (\mathcal{X},\mathcal{X}_0,g,M)$ is said to be an interconnected monotone system (IMS) if each subsystem $\Sigma_i = (\mathcal{X}_i,\mathcal{X}_{0,i},\mathcal{W}_i,g_i)$ is monotone, for $i \in \{1,\ldots,N\}$.
\end{definition}
\smallskip

In Definition~\ref{def:IMS}, monotonicity is imposed at the subsystem level, without requiring the overall interconnected system to be monotone. Many natural and engineered systems admit such a structure, including traffic networks~\cite{SC-MA:15} and biological systems~\cite{EDS:07}. Moreover, a broad class of non-monotone systems can be represented as interconnections of monotone subsystems through a suitable choice of the interconnection matrix. A notable class consists of systems with sign-stable state-transition maps, which encompass many real-world non-monotone dynamics, such as population models~\cite{de2005predator} and gene regulatory networks~\cite{enciso2006nonmonotone}.
\begin{theorem}[Sign-stable Systems]
Consider a discrete-time dynamical system $\Sigma = (\real^n,\prod_{i=1}^{n}\mathcal{X}_{0,i},f)$ with $\mathcal{X}_{0,i} \subseteq \real$, where the state-transition map $f:\real^n \to \real^n$ is sign-stable on $\real^n$. Then $\Sigma$ can be represented as an interconnected monotone system. 
\end{theorem}
\begin{proof}
Since $f$ is sign-stable, for every $i,j\in\{1,\ldots,n\}$, one can define 
\[
s_{ij}:=
\begin{cases}
1, & \displaystyle\frac{\partial f_i}{\partial x_j}(x)\geq0,
     \quad \forall x\in\mathbb{R}^n,\\[1mm]
-1, & \displaystyle\frac{\partial f_i}{\partial x_j}(x)<0,
     \quad \forall x\in\mathbb{R}^n,
\end{cases}
\]
For every $i\in \{1,\ldots,n\}$, we set $S_i:=\operatorname{diag}(s_{i1},\ldots,s_{in})$ and consider $N=n$ scalar-state subsystems $\Sigma_i = (\mathcal{X}_i,\mathcal{X}_{0,i},\mathcal{W}_i,g_i)$ with
$\mathcal{X}_i=\mathbb{R}$, $\mathcal{W}_i=\mathbb{R}^n$, and
\[
g_i(x_i,w_i):=f_i(S_iw_i).
\]
The map $g_i$ is independent of $x_i$, and, for every
$j\in\{1,\ldots,n\}$,
\[
\frac{\partial g_i}{\partial x_i}=0,
\qquad
\frac{\partial g_i}{\partial(w_i)_j}
=
s_{ij}
\frac{\partial f_i}{\partial x_j}(S_iw_i)
\geq0.
\]
Therefore, every subsystem $\Sigma_i$ is monotone. Define
\[
M:=
\begin{bmatrix}
S_1\\ \vdots\\ S_n
\end{bmatrix}
\in\mathbb{R}^{n^2\times n}.
\]
Then $(Mx)_i=S_ix$. Since $S_i^2=I_n$, we have
\[
g_i(x_i,(Mx)_i)
=
f_i(S_iS_ix)
=
f_i(x),
\qquad i\in\{1,\ldots,n\}.
\]
Thus, defining
\begin{align*}
g(x) = (g_1(x_1,(Mx)_1)^{\top},\ldots, g_n(x_n,(Mx)_n)^{\top})^{\top},
\end{align*}
the tuple $(\real^n,\prod_{i=1}^{n}\mathcal{X}_{0,i},g,M)$ represents $\Sigma$ as an interconnected system with monotone subsystems $\Sigma_i$.
\end{proof}

\subsection{Barrier Certificate for Safety Verification}

Ensuring safety is a fundamental objective in the analysis and control of safety-critical systems. In this section, we present a set-theoretic framework for studying safety of dynamical systems.

\begin{definition}[Safety of Systems]
\label{def:safety}
Consider a discrete-time dynamical system (dtDS)
$\Sigma = (\mathcal{X}, \mathcal{X}_{0}, \mathcal{W}, f)$ 
and a set of unsafe states $\mathcal{X}_{u} \subseteq \mathcal{X}$. 
The system $\Sigma$ is \emph{safe} with respect to 
$\mathcal{X}_{u}$ if, for every initial state $x_0 \in \mathcal{X}_{0}$ 
and every input sequence $\{w_t\}_{t=0}^{\infty}\in \mathcal{W}
$, 
the corresponding trajectory $\{x_t\}_{t=0}^{\infty}$ satisfies
\[
x_t \notin \mathcal{X}_{u}, \quad \forall t \ge 0.
\]
\end{definition}
\smallskip

One of the standard approaches for ensuring safety of dynamical systems is the use of barrier certificates. A barrier certificate provides a functional characterization of the separation between safe and unsafe regions of the state space by enforcing an inductive condition that guarantees system trajectories remain within the safe set for all time~\cite{prajna2004safety}.

\begin{definition}[Barrier certificate]
\label{def:barrier}
Consider a discrete-time dynamical system
$\Sigma = (\mathcal{X}, \mathcal{X}_{0}, \mathcal{W}, f)$ 
and a set of unsafe states $\mathcal{X}_{u} \subseteq \mathcal{X}$. 
A function $\mathcal{B}:\mathcal{X} \to \mathbb{R}$ is a
\emph{barrier certificate} for $\Sigma$ with respect to 
$\mathcal{X}_{u}$ if there exist scalars $\gamma, \eta, \delta \in \mathbb{R}$ 
with $\gamma \le \delta \le  \eta$ such that
\begin{align}
&\mathcal{B}(x) \le \gamma, && \forall x \in \mathcal{X}_{0}, \label{eq:bc_initial}\\
&\mathcal{B}(x) > \eta, && \forall x \in \mathcal{X}_{u}, \label{eq:bc_unsafe}\\
&\mathcal{B}(x) \le \delta \implies \mathcal{B}(f(x,w)) \le \delta,&&\forall x \in \mathcal{X},\ \forall w \in \mathcal{W}.\label{eq:bc_invariance}
\end{align}
\end{definition}
\smallskip

The existence of a barrier certificate is a necessary and sufficient condition for safety of the system~\cite{murali2023co}.

\subsection{Problem Statement}

In many real-world applications, the components of large-scale interconnected systems are not fully known and are accessible only through simulations or sampled data. This lack of explicit models poses challenges for formal safety verification. Motivated by this setting, we aim to verify the safety of interconnected monotone systems with unknown components using only sampled data from each subsystem, without relying on a centralized model of the overall system.

\begin{problem}[Safety verification of IMS]\label{p1}
Consider an IMS $\Sigma=(\mathcal{X},\mathcal{X}_0,g,M)$, consisting of $N$ subsystems $\Sigma_i = (\mathcal{X}_i,\mathcal{X}_{0,i},\mathcal{W}_i,g_i)$ with an unsafe set 
$\mathcal{X}_u = \prod_{i=1}^{N} \mathcal{X}_{u,i} \subseteq \mathcal{X}$. We assume that the interconnection matrix $M$ is known. For each subsystem $i \in \{1,\dots,N\}$, the local state-transition map $g_i$ is unknown. The goal is to determine, using finitely many queries of the state-transition maps $g_i$, whether the system $\Sigma$ is safe with respect to the unsafe set $\mathcal{X}_u$.
\end{problem}

In Problem~\ref{p1}, we assume that the interconnection matrix $M$ is known, while the subsystems are unknown but monotone. The first assumption is typically justified by knowledge of the system’s interconnection topology, whereas the second assumption is often supported by prior knowledge of the subsystems’ physical or structural properties. Our goal is to verify the safety of the overall system in a decentralized manner by constructing local barrier certificates for each subsystem. This setting is particularly relevant in real-world applications, where subsystems may be geographically distributed or physically separated, and centralized data collection may be impractical. Moreover, this approach reduces computational complexity and memory requirements by decomposing a global safety verification problem into smaller local subproblems.

Several methods have been proposed for data-driven safety verification of interconnected systems~\cite{noroozi2021data,samariTAC2026}. However, these approaches either lack formal guarantees or, when such guarantees are provided, they are often data-intensive in enforcing local conditions and require Lipschitz bounds on both the system and the barrier function to construct a dense discretization of the state space. Motivated by the limitations of existing approaches, we propose a data-driven, sample-efficient, and scalable framework for the safety verification of interconnected monotone systems.

\section{Safety of Interconnected Monotone Systems}

 Leveraging the interconnected structure of the system, our approach decomposes safety verification into a \textit{local} function search and a \textit{global} certificate construction. At the local level, we use interval analysis to construct suitable data-driven certificates for subsystems. At the global level, we leverage knowledge of the system topology to compose these local certificates and certify the safety of the overall system.

\subsection{Local Interval-Barrier Certificates}

We begin by focusing on local certificates, introducing interval-barrier functions as subsystem-level constructs. Leveraging monotonicity, these certificates can be efficiently constructed from sampled data via interval analysis. Before formally defining interval-barrier certificates, we first describe the state and input set partitioning used in our framework. Suppose that $\Sigma=(\mathcal{X},\mathcal{X}_0,g,M)$ is an interconnected monotone system. For each subsystem $\Sigma_i=(\mathcal{X}_i,\mathcal{X}_{0,i},\mathcal{W}_i,g_i)$, we partition the state space $\mathcal{X}_i$ and the internal input space $\mathcal{W}_i$ into finitely many hyper-rectangular cells. These partitions are represented by the collections 
$\{[\underline{x}_i^j,\overline{x}_i^j]\}_{j \in \mathcal{P}_i}$ and 
$\{[\underline{w}_i^k,\overline{w}_i^k]\}_{k \in \mathcal{Q}_i}$, respectively, 
where $\mathcal{P}_i$ and $\mathcal{Q}_i$ denote the corresponding finite 
index sets. 
For each subsystem, the state and internal-input cells cover their
respective domains:
$\mathcal{X}_i=\bigcup_{j\in\mathcal{P}_i}
[\underline{x}_i^j,\overline{x}_i^j],
\mathcal{W}_i=\bigcup_{k\in\mathcal{Q}_i}
[\underline{w}_i^k,\overline{w}_i^k]$.
We then choose index sets
$\mathcal{I}_i\subseteq\mathcal{P}_i$ and
$\mathcal{U}_i\subseteq\mathcal{P}_i$ such that the corresponding
cells cover the local initial and unsafe sets, respectively:
\begin{align}
\mathcal{X}_{0,i}
&\subseteq
\bigcup_{j\in\mathcal{I}_i}
[\underline{x}_i^j,\overline{x}_i^j],
&
\mathcal{X}_{u,i}
&\subseteq
\bigcup_{j\in\mathcal{U}_i}
[\underline{x}_i^j,\overline{x}_i^j].
\label{eq:cover}
\end{align}

For a given subsystem partition satisfying~\eqref{eq:cover}, we introduce interval-barrier functions for each subsystem as follows.

\begin{definition}[Local Interval-barrier functions]
\label{def:MQB}
Consider an IMS $\Sigma=(\mathcal{X},\mathcal{X}_0,g,M)$ consisting of $N$ subsystems $\Sigma_i=(\mathcal{X}_i,\mathcal{X}_{0,i},\mathcal{W}_i,g_i)$ with the unsafe set $\mathcal{X}_u=\prod_{i=1}^{N}\mathcal{X}_{u,i}$ and the state set and input set partitions satisfying~\eqref{eq:cover}.
Fix a scalar $\lambda\in\mathbb{R}_{\geq0}$, common to all subsystems. A monotone function $\mathcal{B}_i : \mathcal{X}_i \to \mathbb{R}$ is called 
an \emph{interval-barrier function} for the subsystem $\Sigma_i$ if 
there exist scalars $\gamma_i,\eta_i\in\mathbb{R}$ and a matrix 
$X_i =\begin{bmatrix}
    X^{11}_i & X^{12}_i\\ X^{21}_i & X^{22}_i
\end{bmatrix} \in\mathbb{S}^{m_i+n_i}$ where $X_i^{21}=(X_i^{12})^\top$ such that
\begin{align}
\hspace{-0.3cm}\mathcal{B}_i(\overline{x}_i^j) &\le \gamma_i, 
&& \forall j \in \mathcal{I}_i, \label{eq:ib_initial}\\
\hspace{-0.3cm}\mathcal{B}_i(\underline{x}_i^j) &\ge \eta_i, 
&& \forall j \in \mathcal{U}_i, \label{eq:ib_unsafe}\\
\hspace{-0.3cm}\mathcal{B}_i\big(g_i(\overline{x}_i^j,\overline{w}_i^k)\big) 
&\le \lambda \mathcal{B}_i(\underline{x}_i^j) + \mathcal{Z}_i^{j,k},
&& \forall j \in \mathcal{P}_i,\forall k \in \mathcal{Q}_i. 
\label{eq:ib_dynamics}
\end{align}
Here, $\mathcal{Z}_i^{j,k}$ is defined as
\begin{align*}
\mathcal{Z}_i^{j,k} 
& = [\underline{w}^k_i,\overline{w}^k_i]^{\top}X^{11}_i[\underline{w}^k_i,\overline{w}^k_i] + [\underline{w}^k_i,\overline{w}^k_i]^{\top}X^{12}_i[\underline{x}^j_i,\overline{x}^j_i] \\ & + [\underline{x}^j_i,\overline{x}^j_i]^{\top}X^{21}_i[\underline{w}^k_i,\overline{w}^k_i] + [\underline{x}^j_i,\overline{x}^j_i]^{\top}X^{22}_i[\underline{x}^j_i,\overline{x}^j_i],
\end{align*}
where, for hyper-rectangles
$[\underline{a},\overline{a}]\subseteq\mathbb{R}^p$ and
$[\underline{b},\overline{b}]\subseteq\mathbb{R}^q$, and a matrix
$X\in\mathbb{R}^{p\times q}$, we define
\begin{align}\label{eq:bigformula}
[\underline{a},\overline{a}]^\top
X[\underline{b},\overline{b}]
:=
\sum_{r=1}^{p}\sum_{s=1}^{q}
\min\big\{&
X_{rs}\underline{a}_r\underline{b}_s,\,
X_{rs}\underline{a}_r\overline{b}_s,\\
&
X_{rs}\overline{a}_r\underline{b}_s,\,
X_{rs}\overline{a}_r\overline{b}_s
\big\}.
\end{align}
\end{definition}
\smallskip

The definition of local interval-barrier functions is inherently data-driven, enabling their computation without explicit knowledge of the state-transition map $g_i$, relying instead only on sampled data at the boundary points of the partition of $\mathcal{X}_i$ and $\mathcal{W}_i$. In Section~\ref{sec:nn}, we employ neural networks to learn local interval-barrier functions.

\subsection{Safety Verification via Interval-Barrier functions}

In this section, we use the interconnection structure of the system to combine the local interval-barrier functions constructed for each subsystem and synthesize a barrier certificate for the overall system. Our approach is inherently data-driven and compositional, rather than relying on the construction of a centralized barrier certificate.

\begin{theorem}[Compositional barrier certificates]
\label{th:CMB}
Consider the IMS $\Sigma=(\mathcal{X},\mathcal{X}_0,g,M)$ 
consisting of $N$ subsystems $\Sigma_i=(\mathcal{X}_i,\mathcal{X}_{0,i},\mathcal{W}_i,g_i)$ with state set and input set partitions satisfying~\eqref{eq:cover}. Suppose that, for each subsystem 
$\Sigma_i$, there exist an interval-barrier function 
$\mathcal{B}_i:\mathcal{X}_i\to \real$ and 
a matrix $X_i\in\mathbb{S}^{m_i+n_i}$ satisfying~\eqref{eq:ib_initial}--\eqref{eq:ib_dynamics}.
Define the structured symmetric matrix
\[
\Xi :=
\begin{bmatrix}
\operatorname{diag}(X_1^{11},\dots,X_N^{11}) & \operatorname{diag}(X_1^{12},\dots,X_N^{12})\\
\operatorname{diag}(X_1^{21},\dots,X_N^{21}) & \operatorname{diag}(X_1^{22},\dots,X_N^{22})
\end{bmatrix}.
\]
If the following global condition holds:
\begin{equation}
\label{eq:big_lmii}
\Delta \;:=\;
\begin{bmatrix} M \\ I_n \end{bmatrix}^{\!\top}
\Xi
\begin{bmatrix} M \\ I_n \end{bmatrix}
\;\preceq\; 0,
\;\;\;
\sum_{i=1}^{N}\gamma_i \le 0,
\;
\sum_{i=1}^{N}\eta_i > 0,
\end{equation}
then the function
$\mathcal{B}(x) = \sum_{i=1}^{N} \mathcal{B}_i(x_i)$
is a barrier certificate for the overall system $\Sigma$, 
and hence guarantees its safety.
\end{theorem}

\begin{proof}
We verify that the function
$\mathcal{B}(x)=\sum_{i=1}^N \mathcal{B}_i(x_i)$ satisfies conditions~\eqref{eq:bc_initial},~\eqref{eq:bc_unsafe}, and~\eqref{eq:bc_invariance}. Regarding condition~\eqref{eq:bc_initial}, let $x\in \mathcal{X}_0=\prod_{i=1}^N \mathcal{X}_{0,i}$. For each $i$, by the
covering property \eqref{eq:cover} there exists $j\in\mathcal{I}_i$ such that
$x_i \le \overline{x}_i^j$. Since $\mathcal{B}_i$ is monotone, we have
$\mathcal{B}_i(x_i)\le \mathcal{B}_i(\overline{x}_i^j)\le \gamma_i$.
Summing over $i$ yields $\mathcal{B}(x)=\sum_{i=1}^N \mathcal{B}_i(x_i)\le \sum_{i=1}^N \gamma_i \le 0$, for every $x\in\mathcal{X}_0$. Regarding condition~\eqref{eq:bc_unsafe}, let $x\in \mathcal{X}_u=\prod_{i=1}^N \mathcal{X}_{u,i}$. For each $i$, by
\eqref{eq:cover} there exists $j\in\mathcal{U}_i$ such that
$\underline{x}_i^j \le x_i$. Monotonicity of $\mathcal{B}_i$ gives $\mathcal{B}_i(x_i)\ge \mathcal{B}_i(\underline{x}_i^j)\ge \eta_i$. 
Summing over $i$ yields
\[
\mathcal{B}(x)=\sum_{i=1}^N \mathcal{B}_i(x_i)\ge \sum_{i=1}^N \eta_i > 0,
\qquad \forall x\in\mathcal{X}_u.
\]
Regarding condition~\eqref{eq:bc_invariance}, fix $x\in\mathcal{X}$ and let $w=Mx$ be the interconnection input. For each
subsystem $i$, select indices $j\in\mathcal{P}_i$ and $k\in\mathcal{Q}_i$ such
that $x_i\in[\underline{x}_i^j,\overline{x}_i^j]$ and $w_i\in[\underline{w}_i^k,\overline{w}_i^k]$. By monotonicity of $g_i$ and $\mathcal{B}_i$,
$\mathcal{B}_i(g_i(x_i,w_i))
\le 
\mathcal{B}_i\!\big(g_i(\overline{x}_i^j,\overline{w}_i^k)\big)$.
Applying \eqref{eq:ib_dynamics} and using monotonicity of $\mathcal{B}_i$ again (so that
$\mathcal{B}_i(\underline{x}_i^j)\le \mathcal{B}_i(x_i)$) yields
$\mathcal{B}_i(g_i(x_i,w_i))
\le 
\lambda\,\mathcal{B}_i(x_i) + \mathcal{Z}_i^{j,k}$.
Summing over $i$ gives
\begin{equation}\label{eq:sum_step_new}
\mathcal{B}(g(x))
=
\sum_{i=1}^N \mathcal{B}_i(g_i(x_i,w_i))
\le
\lambda\,\mathcal{B}(x) + \sum_{i=1}^N \mathcal{Z}_i^{j,k}.
\end{equation}
We now show that if the global condition~\eqref{eq:big_lmii} holds then $\sum_{i=1}^N \mathcal{Z}_i^{j,k} \le 0$. First, we note that by construction~\eqref{eq:bigformula}, we have $[\underline{a},\overline{a}]^\top
X[\underline{b},\overline{b}]
\leq a^\top Xb$, for all $a\in[\underline{a},\overline{a}]$ and
$b\in[\underline{b},\overline{b}]$. This implies that, for every $x_i\in[\underline{x}_i^j,\overline{x}_i^j]$ and $w_i\in[\underline{w}_i^k,\overline{w}_i^k]$,
\begin{align}\label{eq:Qminineq}
[\underline{w}^k_i,\overline{w}^k_i]^{\top}X^{11}_i[\underline{w}^k_i,\overline{w}^k_i] &\le  w_i^{\top}X^{11}_i w_i,\nonumber\\
    [\underline{w}^k_i,\overline{w}^k_i]^{\top}X^{12}_i[\underline{x}^j_i,\overline{x}^j_i] &\le  w_i^{\top}X^{12}_i x_i,\nonumber\\
    [\underline{x}^j_i,\overline{x}^j_i]^{\top}X^{21}_i[\underline{w}^k_i,\overline{w}^k_i] &\le  x_i^{\top}X^{21}_i w_i,\nonumber\\ 
    [\underline{x}^j_i,\overline{x}^j_i]^{\top}X^{22}_i[\underline{x}^j_i,\overline{x}^j_i]& \le x_i^{\top}X^{22}_i x_i. 
 \end{align}
For each $i$, define $z_i := \begin{bmatrix} w_i \\ x_i \end{bmatrix}$. Summing the inequalities in~\eqref{eq:Qminineq}, we obtain 
\begin{align}\label{eq:Zi_bound}
\mathcal{Z}_i^{j,k} &\le w_i^\top X_i^{11} w_i
+ w_i^\top X_i^{12} x_i
+ x_i^\top X_i^{21} w_i
+ x_i^\top X_i^{22} x_i \nonumber\\ & =  z_i^\top X_i z_i .
\end{align}
Summing \eqref{eq:Zi_bound} over $i$ yields
\begin{equation}\label{eq:sumZi_bound}
\sum_{i=1}^N \mathcal{Z}_i^{j,k} \le \sum_{i=1}^N z_i^\top X_i z_i.
\end{equation}
Next, stack the interconnection variables as $w = (w_1^{\top},\ldots,w_N^{\top})^{\top}$, $x = (x_1^{\top},\ldots,x_N^{\top})^{\top}$, and $z=\begin{bmatrix}w\\x\end{bmatrix}$. 
Using the block partition $X_i=\begin{bmatrix}X_i^{11}&X_i^{12}\\X_i^{21}&X_i^{22}\end{bmatrix}$
and the definition of $\Xi$, we have the identity
\begin{equation}\label{eq:sum_quad_to_big}
\sum_{i=1}^N z_i^\top X_i z_i
=
z^\top \Xi\, z.
\end{equation}

Since $w=Mx$, we have $z=\begin{bmatrix}Mx\\x\end{bmatrix}$, and therefore
\[
z^\top \Xi z
=
x^\top 
\begin{bmatrix} M \\ I_n \end{bmatrix}^{\!\top}
\Xi
\begin{bmatrix} M \\ I_n \end{bmatrix}
x
=
x^\top \Delta x.
\]
Because $\Delta \preceq 0$ by \eqref{eq:big_lmii}, it follows that
$z^\top \Xi z \le 0$, and hence, using
\eqref{eq:sumZi_bound}--\eqref{eq:sum_quad_to_big}, we get $\sum_{i=1}^N \mathcal{Z}_i^{j,k} \le 0$. 
Substituting this into \eqref{eq:sum_step_new} yields
$\mathcal{B}(g(x)) \le \lambda\,\mathcal{B}(x)$. Since $\lambda\ge0$, $\mathcal{B}(x)\le0$ implies
$\mathcal{B}(g(x))\le0$. Together with the initial- and
unsafe-set bounds above, Definition~\ref{def:barrier} holds
with the global parameters $\gamma=\delta=\eta=0$.
\end{proof}

\begin{remark} The following remarks are in order.

\noindent(\textit{Compositional approach}):  Theorem~\ref{th:CMB} provides a data-driven compositional framework for safety verification by (i) local construction of interval-barrier functions, and (ii) global verification of the LMI~\eqref{eq:big_lmii} by composing local interval-barrier certificates using the system topology. 
Importantly, the local conditions use only transition data from the corresponding
subsystem, while compositional compatibility is checked separately
through~\eqref{eq:big_lmii}.
    
    \noindent(\textit{Monotonicity of subsystems}): Monotonicity is central to the proof of Theorem~\ref{th:CMB}. It enables the use of interval analysis to ensure that interval-barrier functions satisfying~\eqref{eq:ib_initial}–\eqref{eq:ib_dynamics}—validated on boundary data—extend to valid barrier certificates over the entire state space.  
     
\noindent\textit{(Computational cost of the global condition):} Theorem~\ref{th:CMB} requires the matrices $X_i$ associated with the local interval-barrier functions in Definition~\ref{def:MQB} to satisfy the global LMI condition~\eqref{eq:big_lmii}. For given matrices $X_i$, forming and checking this matrix inequality requires $\mathcal{O}((m+n)^2n+(m+n)n^2+n^3)$ operations. In the next subsection, we show that, for scalar subsystems with
a symmetric interconnection matrix and a common matrix $X_i=X$, the global matrix inequality is equivalent to a set of scalar inequalities. Once the eigenvalues of the interconnection matrix are precomputed, these inequalities can be checked in $\mathcal{O}(N)$ operations for each candidate $X$ and incorporated into the learning framework of Section~IV.
    
\end{remark}

\subsection{LMI-Free Safety Verification}

Despite its compositional structure, safety verification via Theorem~\ref{th:CMB} may still incur significant computational overhead for large-scale networks due to the global LMI condition~\eqref{eq:big_lmii}. To address this limitation, in this section we focus on a particular class of interconnected systems for which the global condition can be reduced to an equivalent characterization in terms of low-dimensional inequalities. 


\begin{theorem}[LMI-free safety verification]
\label{thm:eig_reduction_lmi}
Consider an IMS $\Sigma=(\mathcal{X},\mathcal{X}_0,g,M)$
consisting of $N$ scalar subsystems
$\Sigma_i=(\mathcal{X}_i,\mathcal{X}_{0,i},
\mathcal{W}_i,g_i)$, i.e.,
$n_i=m_i=1$ for all $i\in\{1,\ldots,N\}$, with state
and input partitions satisfying~\eqref{eq:cover}.
Suppose that the interconnection matrix
$M\in\mathbb{S}^{N}$ is symmetric and that there exist
interval-barrier functions
$\mathcal{B}_i:\mathcal{X}_i\to\mathbb{R}$ satisfying
conditions~\eqref{eq:ib_initial}--\eqref{eq:ib_dynamics}
with the same matrix
$X_i=X
\in\mathbb{S}^{2}$
for all $i\in\{1,\ldots,N\}$.
If, for every eigenvalue $\mu_k$ of $M$,
$k\in\{1,\ldots,N\}$,
\begin{equation}
\label{eq:scalar_ineq}
\begin{bmatrix}\mu_k\\1\end{bmatrix}^{\!\top}
X
\begin{bmatrix}\mu_k\\1\end{bmatrix}
\le0,
\qquad
\sum_{i=1}^{N}\gamma_i\le0,
\qquad
\sum_{i=1}^{N}\eta_i>0,
\end{equation}
then
$\mathcal{B}(x)=\sum_{i=1}^{N}\mathcal{B}_i(x_i)$
is a barrier certificate for $\Sigma$, and hence the system
is safe.
\end{theorem}

\begin{proof}
It suffices to show that, under the assumptions of the
theorem, the matrix inequality in~\eqref{eq:big_lmii} is
equivalent to the eigenvalue inequalities
in~\eqref{eq:scalar_ineq}. Since $X_i=X\in\mathbb{S}^2$
for all $i\in\{1,\ldots,N\}$, the structured matrix
$\Xi$ in~\eqref{eq:big_lmii} satisfies
$\Xi=X\otimes I_N$. Therefore, the global condition becomes
\begin{equation}
\label{eq:big_lmi}
\Delta :=
\begin{bmatrix}M\\I_N\end{bmatrix}^{\!\top}\!\!\!
(X\otimes I_N)
\begin{bmatrix}M\\I_N\end{bmatrix}
\preceq0,
\quad
\sum_{i=1}^{N}\gamma_i\le0,
\quad
\sum_{i=1}^{N}\eta_i>0.
\end{equation}
Since $M\in\mathbb{S}^{N}$ is real and symmetric, it admits a spectral decomposition $M=Q\Lambda Q^\top,\; Q^\top Q=I_N$,
where
$\Lambda=\operatorname{diag}(\mu_1,\ldots,\mu_N)$
contains the eigenvalues of $M$.
Define $S:=\begin{bmatrix}M\\I_N\end{bmatrix}
\in\mathbb{R}^{2N\times N}$.
Using $M=Q\Lambda Q^\top$, we can factor
\begin{equation}
\label{eq:S_factor}
S
=
\begin{bmatrix}Q\Lambda Q^\top\\ I_N\end{bmatrix}
=
\underbrace{\begin{bmatrix}Q & 0\\ 0 & Q\end{bmatrix}}_{=: \mathcal Q \in \mathbb{R}^{2N\times 2N}}
\begin{bmatrix}\Lambda Q^\top\\ Q^\top\end{bmatrix}.
\end{equation}
Since $Q$ is orthogonal, so is $\mathcal Q$, i.e., $\mathcal Q^\top\mathcal Q=I_{2N}$.
Next, we claim that $X\otimes I_N$ is invariant under congruence by $\mathcal Q$:
\begin{equation}
\label{eq:invariance}
\mathcal Q^\top (X\otimes I_N) \, \mathcal Q = X\otimes I_N.
\end{equation}
Indeed, noting that $\mathcal Q = I_2\otimes Q$, the mixed-product property of the Kronecker product yields
\begin{align*}
 \mathcal Q^\top (X\otimes I_N\,) \mathcal Q
& = (I_2\otimes Q^\top)(X\otimes I_N)(I_2\otimes Q) \\ 
& = X\otimes(Q^\top I_N Q)
= X\otimes I_N, 
\end{align*}
which proves \eqref{eq:invariance}.
Substituting the factorization \eqref{eq:S_factor} into \eqref{eq:big_lmi} and using \eqref{eq:invariance}, we obtain
\begin{align*}
\Delta
&= S^\top (X\otimes I_N) S
=
\begin{bmatrix}\Lambda Q^\top\\ Q^\top\end{bmatrix}^{\!\top}
\mathcal Q^\top (X\otimes I_N) \mathcal Q
\begin{bmatrix}\Lambda Q^\top\\ Q^\top\end{bmatrix} \\&=
\begin{bmatrix}\Lambda Q^\top\\ Q^\top\end{bmatrix}^{\!\top}
(X\otimes I_N)
\begin{bmatrix}\Lambda Q^\top\\ Q^\top\end{bmatrix}.
\end{align*}
Factoring $Q^\top$ on the right gives $\begin{bmatrix}\Lambda Q^\top\\ Q^\top\end{bmatrix}
=
\begin{bmatrix}\Lambda\\ I_N\end{bmatrix}Q^\top
$. Therefore
\begin{equation}
\label{eq:Delta_congruence}
\Delta
=
Q\left(
\begin{bmatrix}\Lambda\\ I_N\end{bmatrix}^{\!\top}
(X\otimes I_N)
\begin{bmatrix}\Lambda\\ I_N\end{bmatrix}
\right)Q^\top
=: Q\,\widetilde\Delta\,Q^\top.
\end{equation}
Since congruence with an orthogonal matrix preserves semidefiniteness, we have
$\Delta\preceq 0 \iff \widetilde\Delta\preceq 0$.
Because $\Lambda$ is diagonal and $X\otimes I_N$ consists of $2\times2$ scalar blocks multiplied by $I_N$, $\widetilde{\Delta}$ is diagonal, with
\[
\left(\widetilde{\Delta}\right)_{kk}
=
\begin{bmatrix}\mu_k\\1\end{bmatrix}^{\!\top}
X
\begin{bmatrix}\mu_k\\1\end{bmatrix},
\qquad k\in\{1,\ldots,N\}.
\]
Hence,
$\widetilde\Delta \preceq 0
\iff
\begin{bmatrix}\mu_k\\1\end{bmatrix}^{\!\top}X\begin{bmatrix}\mu_k\\1\end{bmatrix}\le 0,\ \forall k\in\{1,\ldots,N \}.$
This proves the equivalence between the matrix inequality in~\eqref{eq:big_lmi} and the eigenvalue inequalities in~\eqref{eq:scalar_ineq}.
\end{proof}
\smallskip

\begin{remark} The following remarks are in order.

\noindent\textit{(LMI-free safety verification).} This reformulation eliminates the need to explicitly enforce the global LMI and instead incorporates the global condition into local training objectives. In Section~\ref{sec:nn}, we show how this condition can be embedded into the loss function of neural networks for learning interval-barrier functions, thereby guaranteeing safety without explicitly enforcing the global LMI condition.

\noindent\textit{(Computational time).}
Theorem~\ref{thm:eig_reduction_lmi} shows that, for scalar subsystems with a symmetric interconnection matrix, the global LMI condition is equivalent to $N$ scalar
inequalities. If the eigen-decomposition of $M$ is precomputed and the common matrix $X$ is fixed, the cost of checking these $N$ scalar inequalities is $\mathcal{O}(N)$.
\end{remark}
\bigskip

Theorems~\ref{th:CMB} and~\ref{thm:eig_reduction_lmi} establish safety guarantees for interconnected monotone systems given suitable local interval-barrier functions and known interconnection matrix, but do not address their construction. In the next section, we develop a neural network–based approach to learn such functions.

\section{Monotone Neural Interval-Barriers}\label{sec:nn}

A common approach for learning data-driven local monotone interval-barrier functions satisfying \eqref{eq:ib_initial}–\eqref{eq:ib_dynamics} while enforcing the global condition \eqref{eq:big_lmii} is to use template-based parameterizations, such as polynomials. However, such methods face a trade-off between expressiveness and tractability and scale poorly with system dimension due to the combinatorial growth in decision variables. Moreover, enforcing monotonicity and compatibility with global conditions such as \eqref{eq:big_lmii} is nontrivial.
In this section, we instead employ monotone neural networks, which provide a flexible and scalable framework for learning interval-barrier certificates while incorporating the required structural constraints.

\subsection{Monotone Neural Networks}
The design of monotone neural networks has been extensively studied, with several architectures demonstrating strong performance across a range of applications~\cite{sill_1998,HD-MV:10}. We adopt a standard feedforward architecture to model such networks. Specifically, consider a neural network $\mathcal{N}:\mathbb{R}^{n_0}\to\mathbb{R}$ with $L$ hidden layers followed by a scalar affine output layer. Let $W^\ell$ and $b^\ell$, for $\ell=0,\ldots,L-1$, denote the weight matrix and bias vector of the hidden-layer transformations, respectively, and let $W^L$ and $b^L$ denote the weight and bias of the scalar output layer. For an input $x\in\mathbb{R}^{n_0}$, the network output is computed recursively as
\begin{align}
y^0 &= x, \label{eq:monNN1_new}\\
y^{\ell+1} &= \sigma\!\left(W^\ell y^\ell+b^\ell\right),
\quad \ell=0,\ldots,L-1, \label{eq:monNN2_new}\\
\mathcal{N}(x) &= W^L y^L+b^L. \label{eq:monNN3_new}
\end{align}
Here, $y^\ell$ denotes the feature vector at layer $\ell$, and the activation function $\sigma$ is applied elementwise. A sufficient condition for network monotonicity is that all weight matrices are entrywise nonnegative, i.e., $W^\ell\geq 0$ for all $\ell=0,\ldots,L$, and that $\sigma$ is nondecreasing~\cite{HD-MV:10}. 

\subsection{Neural Interval-Barriers}

In this section, we use the monotone neural network architecture in~\eqref{eq:monNN1_new}--\eqref{eq:monNN3_new} as a template to learn the local interval-barrier certificates. Building on Theorem~\ref{th:CMB}, we propose the following data-driven formulation for learning monotone neural interval-barrier functions $\mathcal{N}_i: \mathcal{X}_i \to \mathbb{R}$:
\begin{align}
&\mathcal{N}_i(\overline{x}_i^j)  \le \gamma_i
\hspace{-0.2cm}&& \forall j \!\in\! \mathcal{I}_i,\label{eq:n1} \\ 
&\mathcal{N}_i(\underline{x}_i^j) \ge \eta_i  
\hspace{-0.2cm}&& \forall j \!\in\! \mathcal{U}_i,\label{eq:n2} \\ 
&\mathcal{N}_i\big(g_i(\overline{x}_i^j,\overline{w}_i^k)\big)
\le \lambda \mathcal{N}_i(\underline{x}_i^j) \!+\! \mathcal{Z}_i^{j,k},
\hspace{-0.2cm}&& \forall j \!\in\! \mathcal{P}_i,\ \forall k \!\in\! \mathcal{Q}_i,\label{eq:n3} \\ 
&W_i^{\ell} \ge 0,
\hspace{-0.2cm}&& \ell \in \{0,\ldots,L\}.\label{eq:n4}
\end{align}

Condition~\eqref{eq:n4} enforces the monotonicity of the neural network $\mathcal{N}_i$, while conditions~\eqref{eq:n1}--\eqref{eq:n3} ensure that $\mathcal{N}_i$ constitutes an interval-barrier function for subsystem $\Sigma_i$. Learning neural networks that satisfy all conditions~\eqref{eq:n1}--\eqref{eq:n4} is a nontrivial task. To address this, we impose condition~\eqref{eq:n4} directly on the
network weights and use the loss functions introduced below to enforce
conditions~\eqref{eq:n1}--\eqref{eq:n3}.

For each subsystem $\Sigma_i$, we train a monotone neural network $\mathcal{N}_i:\mathcal{X}_i\to\mathbb{R}$ to serve as an interval-barrier certificate using a squared ReLU-based loss function. The total loss is defined as $L_i^{\text{total}} = L_i^1 + L_i^2 + L_i^3$,
where each term $L^1_i$, $L^2_i$, and $L^3_i$ corresponds to enforcing one of the interval-barrier conditions~\eqref{eq:n1}--\eqref{eq:n3} as follows:
\begin{align}
\hspace{-0.15cm} L_i^1 &=\!\!\sum_{j \in \mathcal{I}_i}
\!\!\big[\mathrm{ReLU}\big(  \mathcal{N}_i(\overline{x}_i^j) - \gamma_i  \big)\big]^2, \label{eq:loss1}\\
\hspace{-0.15cm}L_i^2 &=\!\!\sum_{j \in \mathcal{U}_i}
\!\!\big[\mathrm{ReLU}\big( \eta_i - \mathcal{N}_i(\underline{x}_i^j) \big)\big]^2, \label{eq:loss2}\\ 
\hspace{-0.15cm}L_i^3 &=\!\!\!\!\!\!\!\!\!\!\sum_{\substack{j \in \mathcal{P}_i,\;k \in \mathcal{Q}_i}}
\!\!\!\!\!\!\!\!\!\!\big[\mathrm{ReLU}\big( \mathcal{N}_i(g_i(\overline{x}_i^j,\overline{w}_i^k))
- \!\lambda \mathcal{N}_i(\underline{x}_i^j)\!- \mathcal{Z}_i^{j,k} \big)\big]^2. \label{eq:loss3}
\end{align}

Condition~\eqref{eq:n4} is enforced directly by constraining the network
weights during training. Consequently, zero total loss guarantees
conditions~\eqref{eq:n1}--\eqref{eq:n3}, which, together with
\eqref{eq:n4}, establish that each $\mathcal{N}_i$ is a monotone
interval-barrier function.
If, in addition, the associated matrices $X_i$ satisfy the global condition~\eqref{eq:big_lmii}—either by construction or via a posteriori verification—then the requirements of Theorem~\ref{th:CMB} are met. Consequently, the sum of the learned local neural functions constitutes a valid barrier certificate for the overall system, thereby guaranteeing safety of the interconnected system.
This result is formally stated in the following theorem.

\begin{theorem}[Neural Interval-barriers]\label{th:NNmonbarr}
Consider an interconnected monotone system $\Sigma=(\mathcal{X},\mathcal{X}_0,g,M)$ consisting of $N$ subsystems $\Sigma_i=(\mathcal{X}_i,\mathcal{X}_{0,i},\mathcal{W}_i,g_i)$ whose state and input partitions satisfying~\eqref{eq:cover}. 
Suppose that, for each subsystem $\Sigma_i$, there exists a monotone neural network candidate $\mathcal{N}_i$ trained using the loss functions defined in~\eqref{eq:loss1}--\eqref{eq:loss3}, with an associated symmetric matrix
$X_i\in\mathbb{S}^{m_i+n_i}$
satisfying the global condition~\eqref{eq:big_lmii}. 
If $L_i^{\mathrm{total}}=0$ for all $i\in\{1,\dots,N\}$, then each $\mathcal{N}_i$ is a valid monotone interval-barrier function for the subsystem $\Sigma_i$. Consequently, the function
$\mathcal{N}(x)=\sum_{i=1}^{N}\mathcal{N}_i(x_i)$
is a valid barrier certificate for the overall system $\Sigma$, and hence $\Sigma$ is safe.
\end{theorem}
\begin{proof}
Assume that $L_i^{\mathrm{total}}=L_i^1+L_i^2+L_i^3=0$ for all 
$i\in\{1,\dots,N\}$. We show that $\mathcal{N}_i$ is a local interval-barrier function for subsystem $\Sigma_i$. Since each loss term is a sum of squared ReLU terms, each term is nonnegative. Hence,
$L_i^{\mathrm{total}}=0
\Longrightarrow
L_i^1=0,\;L_i^2=0,\;L_i^3=0,
\forall i\in\{1,\dots,N\}$.
From $L_i^1=0$ and \eqref{eq:loss1}, every summand must vanish, which implies
$\mathcal{N}_i(\overline{x}_i^j)\le \gamma_i,\; \forall j\in\mathcal{I}_i$.
Thus $\mathcal{N}_i$ satisfies condition~\eqref{eq:ib_initial}. Similarly, from $L_i^2=0$ and
\eqref{eq:loss2}, we obtain
$\mathcal{N}_i(\underline{x}_i^j)\ge \eta_i,\; \forall j\in\mathcal{U}_i$,
and hence $\mathcal{N}_i$ satisfies condition \eqref{eq:ib_unsafe}. Likewise, $L_i^3=0$ together with
\eqref{eq:loss3} yields
$\mathcal{N}_i\big(g_i(\overline{x}_i^j,\overline{w}_i^k)\big)
\le \lambda \mathcal{N}_i(\underline{x}_i^j)+\mathcal{Z}_i^{j,k},
\; \forall j\in\mathcal{P}_i,\ \forall k\in\mathcal{Q}_i$. This means that $\mathcal{N}_i$ satisfies condition~\eqref{eq:ib_dynamics}. Therefore, each $\mathcal{N}_i$ is a valid monotone interval-barrier
function for the subsystem $\Sigma_i$.
If, in addition, the associated matrices $X_i$ satisfy the global
condition~\eqref{eq:big_lmii}, then all assumptions of
Theorem~\ref{th:CMB} are satisfied. Hence,
$\mathcal{N}(x)=\sum_{i=1}^N \mathcal{N}_i(x_i)$
is a barrier certificate for the overall system $\Sigma$, and therefore
$\Sigma$ is safe.
\end{proof}
\smallskip

\begin{remark}
The following remarks are in order.

\noindent\textit{(Loss Function Variants)}: The total loss $L_i^{\text{total}} = L_i^1 + L_i^2 + L_i^3$ can be modified to improve training. For instance, smooth alternatives such as mean squared error (MSE) losses may be employed to enhance gradient behavior, while the proposed loss can be retained as a termination criterion. In our experiments, however, the proposed loss is used for both training and termination, yielding stable optimization.

\noindent\textit{(LMI-Free Formulation)}: The loss function $L^{\mathrm{total}}_i$ applies to the general setting, where the matrices $X_i$ are either specified a priori or verified a posteriori via the global condition~\eqref{eq:big_lmii}. In contrast, under the LMI-free characterization, the global LMI is completely eliminated from the verification process. Instead, the eigenvalue inequalities in~\eqref{eq:scalar_ineq} can be incorporated into the training objective through the shared loss term
$L^{\mathrm{eig}}
=
\sum_{k=1}^{N}
\left[
\mathrm{ReLU}\left(
\mu_k^2X^{11}+2\mu_kX^{12}+X^{22}
\right)
\right]^2$.
This enables the common matrix $X\in\mathbb{S}^2$ to be learned during training without explicitly constructing or checking the global LMI. If $L_i^{\mathrm{total}}=0$ for all $i$, $L^{\mathrm{eig}}=0$, the weight constraints~\eqref{eq:n4} hold, and the scalar sum conditions in~\eqref{eq:scalar_ineq} are satisfied, then the learned neural networks are interval-barrier functions whose sum constitutes a valid barrier certificate for the overall interconnected system.
\end{remark}

Algorithm~\ref{alg:MonNNBarInt} summarizes the general LMI-based
procedure for prescribed matrices $X_i$. In the LMI-free formulation,
the common matrix $X$ can instead be learned by augmenting the
training objective with $L^{\mathrm{eig}}$.

\begin{algorithm}[!t]
\caption{Compositional Neural Safety Verification}
\label{alg:MonNNBarInt}
\textbf{Inputs:} Simulator access to $g_i$, $i\in\{1,\dots,N\}$;
the interconnection matrix $M$; sets $\mathcal{X}_{0,i}$,
$\mathcal{X}_{u,i}$, $\mathcal{X}_i$, and $\mathcal{W}_i$;
index sets $\mathcal{P}_i$, $\mathcal{Q}_i$, $\mathcal{I}_i$,
and $\mathcal{U}_i$; neural-network architectures $\mathcal{N}_i$;
a common scalar $\lambda\geq0$; scalars $\gamma_i,\eta_i$;
fixed matrices $X_i\in\mathbb{S}^{m_i+n_i}$ such that the global
conditions in~\eqref{eq:big_lmii} hold; and a maximum number of
iterations $e_{\max}$.\\
\textbf{Outputs:} Neural interval-barrier functions
$\{\mathcal{N}_i\}_{i=1}^{N}$, or \texttt{FAIL}.
\begin{algorithmic}[1]
\STATE Query $g_i(\overline{x}_i^j,\overline{w}_i^k)$ for all
$i$, $j\in\mathcal{P}_i$, and $k\in\mathcal{Q}_i$
\STATE Initialize the parameters of $\{\mathcal{N}_i\}_{i=1}^{N}$
\STATE Set $e\gets1$
\WHILE{$e\leq e_{\max}$}
    \STATE Update each $\mathcal{N}_i$ by minimizing
    $L_i^{\mathrm{total}}$, subject to
    $W_i^\ell\geq0$ for all $\ell$
    \STATE Compute $L_i^{\mathrm{total}}$ for all $i$
    \IF{$L_i^{\mathrm{total}}=0$ for all $i$}
        \STATE \textbf{return} $\{\mathcal{N}_i\}_{i=1}^{N}$
    \ENDIF
    \STATE $e\gets e+1$
\ENDWHILE
\STATE \textbf{return} \texttt{FAIL}
\end{algorithmic}
\end{algorithm}

\section{Numerical Experiments}
In this section, we evaluate the efficiency of the proposed framework. The learning is performed in a black-box manner using only sampled data. In all numerical experiments, each interval-barrier is learned using a monotone neural network with a single hidden layer of 20 neurons and $\tanh$ activation. Experiments are conducted on a MacBook Pro (Apple M3 Max, 48 GB RAM).

\subsection{Two-Dimensional Interconnected Systems}
\begin{figure}[!b]
 \centering
  \includegraphics[width =0.49\linewidth]{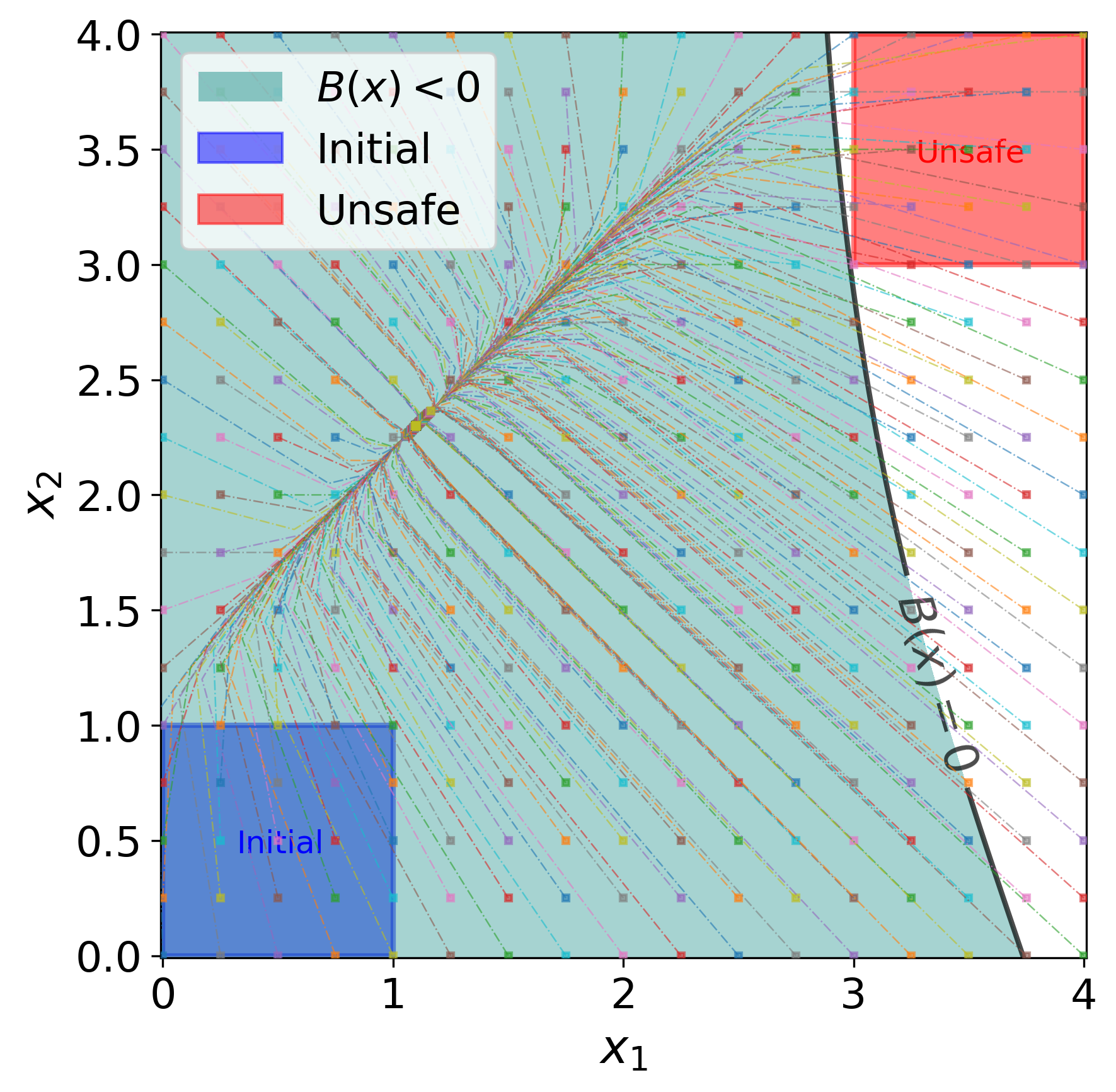}
  \includegraphics[width =0.49\linewidth]{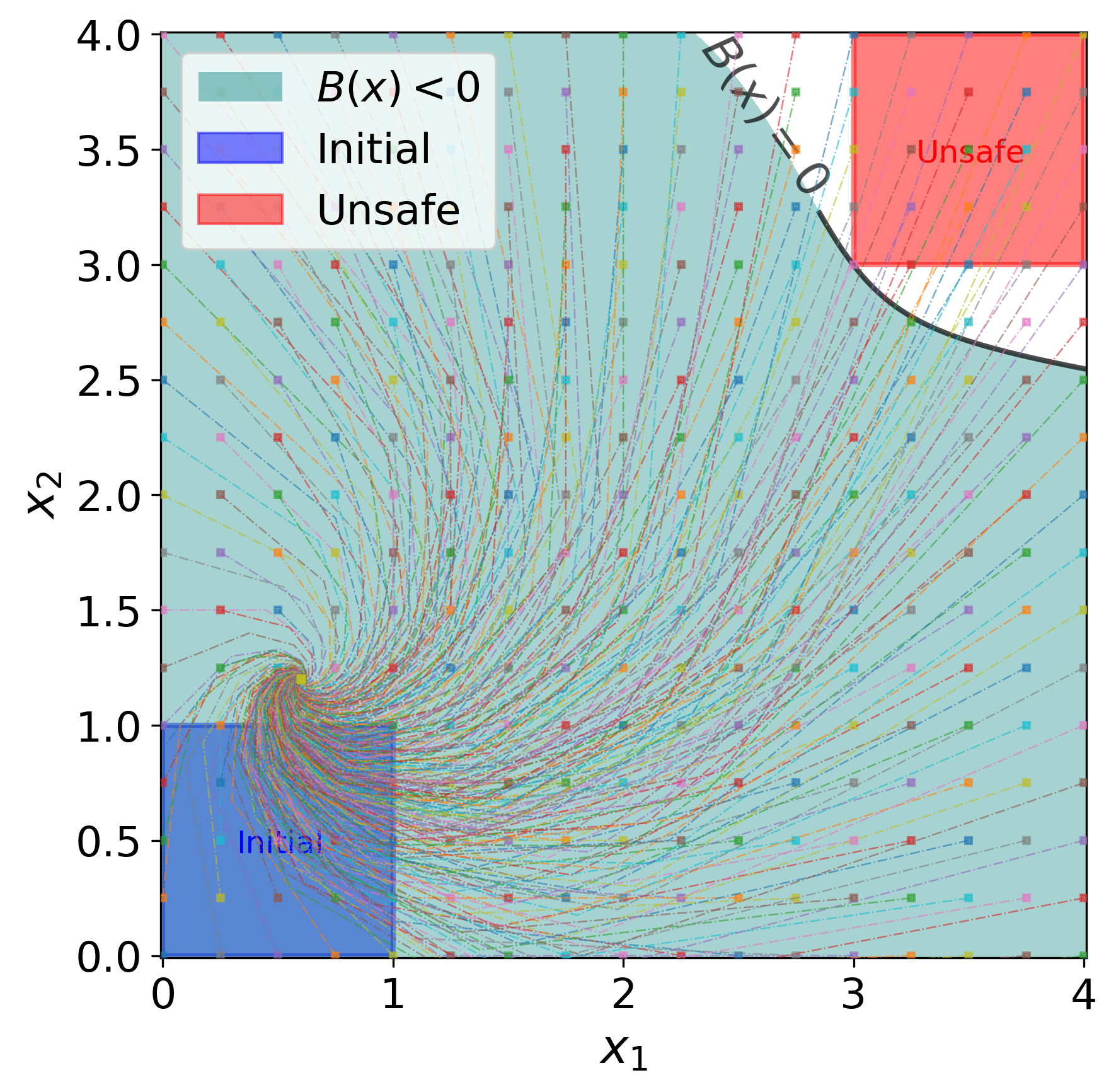}
\caption{Zero sublevel sets of the learned barrier functions. Left: monotone interconnection (Case 1, LMI-free). Right: non-monotone interconnection (Case 2, LMI-based).}
\label{fig:twosys}
 \end{figure}
We consider two representative case studies of interconnected scalar subsystems. In both settings, the state space is $\mathcal{X} = [0,4]^2$, the initial set is $\mathcal{X}_0 = [0,1]^2$, and the unsafe set is $\mathcal{X}_u = [3,4]^2$. The parameter $\lambda$ is set to $10^{-3}$, and the time step is $h = 0.3$. For Case 1 (monotone interconnection), we consider the discrete-time system
\begin{align*}
x_1^{+} &= (1-2h)x_1 + h x_2 - 0.1h,\\
x_2^{+} &= h x_1 + (1-2h)x_2 + 3.5h.
\end{align*} 
This overall system is monotone. By defining $w_1 = x_2$ and $w_2 = x_1$, the system can be represented as an interconnected monotone system with the interconnection matrix 
$M = \left[\begin{smallmatrix} 0 & 1 \\ 1 & 0 \end{smallmatrix}\right]$.
We employ the LMI-free formulation to learn a common matrix $X=X_1=X_2\in\mathbb{S}^2$. The local state and internal input spaces are coarsely partitioned using boundary samples, e.g., $x_i \in [0,4]$, $w_i \in [0,4]$. The learned neural interval-barriers successfully certify safety. The zero sublevel set of the resulting barrier is shown in Fig.~\ref{fig:twosys} (left). For Case 2 (non-monotone interconnection), we consider the discrete-time system
\begin{align*}
x_1^{+} &= (1-2h)x_1 + h x_2,\\
x_2^{+} &= -h x_1 + (1-2h)x_2 + 3.0h.
\end{align*}
The overall system is not monotone. However, by defining $w_1 = x_2$ and $w_2 = -x_1$, the system can be represented as an interconnection of monotone subsystems with
the interconnection matrix  $M = \left[\begin{smallmatrix} 0 & 1 \\ -1 & 0 \end{smallmatrix}\right]$.
Since $M$ is not symmetric, we employ the general LMI condition to compute $X_1$ and $X_2$. The same neural architecture and coarse partitions as in Case 1 are used, with adjusted input bounds reflecting the sign change in $w_2$. The learned interval-barriers again certify safety. The zero sublevel set is shown in Fig.~\ref{fig:twosys} (right). These results highlight the ability of the proposed framework to certify safety using limited data under both monotone and non-monotone interconnections.

\subsection{Gene Regulatory Network}

We consider a gene regulatory network of $N$ subsystems with a cyclic negative feedback interconnection~\cite{GAE-HLS-EDS:06}:
\begin{align*}
x_1^{+} &= x_1 + h\big(-\alpha x_1 + \ell(x_N)\big), \\
x_i^{+} &= x_i + h\big(-\alpha x_i + x_{i-1}\big), \quad i \in\{ 2,\dots,N\},
\end{align*}
where $\alpha > 0$ is the degradation rate and the nonlinear feedback function is defined as
$\ell(x) = a/(1 + k x^q)$,
with parameters $a>0$, $k>0$, and $q \ge 1$. In this work, we consider $\alpha = 10$, $a = 1.0$, $k = 1.0$, and $q = 2$.
Each subsystem is represented in the form $x_i^{+} = g_i(x_i, w_i)$, where the internal inputs are defined as
$w_1 = -x_N,\; w_i = hx_{i-1}, \; i \in\{ 2,\dots,N\}$.
With this representation and step size $h$ chosen such that $h \le 1/\alpha$, each local subsystem is monotone with respect to its state and internal input, while the non-monotonicity of the overall network is captured through the interconnection structure.

\begin{table}[!t]
\centering
\setlength{\tabcolsep}{6.5pt}
\renewcommand{\arraystretch}{1.3}
\begin{tabular}{|c|c|c|c|c|c|}
\hline
Number of subsystems $N$ & $5$ & $10$ & $20$ & $40$  & $50$ \\
\hline
Runtime (min) & $0.4$ & $1.1$ & $1.9$ &  $11.8$ &  $15$  \\
\hline
Number of samples & $310$ & $620$ & $1240$ & $2480$ & $3100$ \\
\hline
\end{tabular}
\caption{Gene regulatory network safety verification.}\label{tab:gene}
\end{table}

The state space is $\mathcal{X} = [0,1]^N$, with initial and unsafe sets $\mathcal{X}_0 = [0,0.2]^N$ and $\mathcal{X}_u = [0.8,1]^N$. The parameter $\lambda$ is set to $10^{-3}$, and a uniform partition with 20 intervals generates boundary samples. We learn one monotone neural barrier for the first subsystem and a shared barrier for the remaining subsystems, together with local matrices $\{X_i\}_{i=1}^N$. Training jointly optimizes these variables under the compositional conditions, with the global constraint enforced via the maximum eigenvalue of $\Delta$. Training terminates once the loss reaches zero, yielding a valid compositional interval-barrier certificate. Table~\ref{tab:gene} summarizes the results, demonstrating scalability and efficient sample usage.

\section{Conclusion}
In this paper, we proposed a compositional framework for the formal safety verification of interconnected monotone subsystems using only finite data, without requiring explicit models of the local dynamics. At the subsystem level, we exploited monotonicity to reduce the search for local barrier certificates to localized boundary conditions, which can be efficiently learned using monotone neural barrier certificates. At the global level, we leveraged system topology to derive an LMI condition that certifies the safety of the overall interconnected system. Experimental results demonstrate the effectiveness and scalability of the proposed framework.


\bibliographystyle{IEEEtran}
\bibliography{ref,SJ}

\end{document}